\documentclass[10pt,draft,reqno]{amsart}
     \makeatletter
     \def\section{\@startsection{section}{1}%
     \z@{.7\linespacing\@plus\linespacing}{.5\linespacing}%
     {\bfseries
     \centering
     }}
     \def\@secnumfont{\bfseries}
     \makeatother
\newtheorem{theorem}{Theorem}[section]
\newtheorem{lemma}[theorem]{Lemma}

\newtheorem{corollary}[theorem]{Corollary}
\theoremstyle{definition}

\theoremstyle{remark}
\newtheorem{remark}[theorem]{Remark}
\numberwithin{equation}{section} 

\begin{document}

\title[Spectral Theorem approach to Quantum Observables]{Spectral Theorem approach
to the  Characteristic Function of Quantum Observables}
\author{Andreas Boukas}
\address{Centro Vito Volterra, Universit\`{a} di Roma Tor Vergata, via Columbia  2, 00133 Roma,
Italy and Graduate School of Mathematics, Hellenic Open
University, Greece} \email{boukas.andreas@ac.eap.gr}
\author{Philip Feinsilver}
\address{Department of Mathematics, Southern Illinois University, Carbondale, Illinois, USA}
\email{pfeinsil@math.siu.edu}

\date{\today}



\begin{abstract} Using the spectral theorem  we
compute the \textit{Quantum Fourier Transform} or \textit{Vacuum
Characteristic Function} $\langle \Phi, e^{itH}\Phi\rangle$ of an
observable $H$ defined as a self-adjoint sum of the generators of
a finite-dimensional Lie algebra, where $\Phi$ is a unit
 vector in a  Hilbert space $\mathcal{H}$. We show
how \textit{Stone's formula} for computing the spectral resolution
of a Hilbert space self-adjoint operator, can serve as an
alternative to the traditional reliance on \textit{splitting} or
\textit{disentanglement} formulas for the operator exponential.
\end{abstract}
\maketitle
\section{Introduction}
The simplest quantum   analogue of a classical probability space
$(\Omega, \sigma, \mu)$  where $\Omega$ is a set, $\sigma$ is a
sigma-algebra of subsets of $\Omega$ and $\mu$ is a measure
defined on $\sigma$  with $\mu(\Omega)=1$, is  a \textit{ finite
dimensional quantum probability space} \cite{Pa} defined as a
triple $(\mathcal{H}, \mathcal{P}(\mathcal{H}), \rho)$ where
$\mathcal{H}$ is a finite dimensional Hilbert space,
$\mathcal{P}(\mathcal{H})$ is the set of projections (called
\textit{events}) $E: \mathcal{H}\to \mathcal{H}$ and $\rho:
\mathcal{H}\to \mathcal{H}$ is a \textit{state} on $\mathcal{H}$,
i.e., a positive operator of unit trace. We call $\rm{tr}\rho E$
the \textit{probability of the event E in the state} $\rho$. For a \textit{quantum observable}  $H$, i.e. for a symmetric or Hermitian matrix $H$, the
\textit{characteristic function} or \textit{Fourier transform} of
$H$ in the state $\rho$ is defined as $\rm{tr}\rho e^{i t H}$. If
$\rho$ is a \textit{pure state} defined in terms of a unit vector
$u$, i.e., if $\rho=|u\rangle\langle u|$ then the characteristic
function of $H$ in the state defined by $u$ is $\langle u,
e^{itH}u \rangle$. By the spectral theorem, if $H=\sum_n\lambda_n E_n$ then
 for every continuous function $\phi: \mathbb{R}\to \mathbb{C}$,
 \[
 \phi(H)=\sum_n\phi(\lambda_n)  E_n .
 \]
  Therefore, for $\phi(H)=e^{itH}$ we have
\[
\langle u, e^{itH}u \rangle=\langle u, \sum_ne^{it\lambda_n}E_n u
\rangle =\sum_ne^{it\lambda_n}\langle u, E_{n}u \rangle ,
\]
where we have assumed that the inner product is linear in the
second and conjugate linear in the first argument.
 If the Hilbert space $\mathcal{H}$ is infinite dimensional then
  the above sums are replaced by spectral integrals with respect to a resolution of the identity $\{E_\lambda\,\,:\,\,\lambda \in\mathbb{R}\}$  and we have the corresponding formulas
\[
H=\int_{\mathbb{R}}\lambda \,dE_\lambda \,\,;\,\,
\phi(H)=\int_{\mathbb{R}}\phi(\lambda)\, dE_\lambda
\]
and
\[
\langle u, e^{itH}u \rangle
=\int_{\mathbb{R}}e^{it\lambda}\,d\langle u, E_{\lambda}u \rangle\
.
\]

For compact self-adjoint operators $H$, the above spectral
integrals are reduced to finite or infinite sums over the nonzero
eigenvalues of $H$, see e.g. \cite{Taylor}, Theorem 4.2.

The above probabilistic interpretation is based on Bochner's
theorem (see \cite{Yosida} p. 346) which states that a
\textit{positive definite} continuous function $f:
\mathbb{R}\mapsto \mathbb{C}$ , i.e., a continuous function $f$
such that
\[
 \int_\mathbb{R}\int_\mathbb{R}f(t-s) \phi (t)\bar  \phi
(s)\,dt\,ds\geq 0,
\]
for every continuous function $ \phi : \mathbb{R}\mapsto
\mathbb{C}$ with compact support, can be represented as
\[
f(t)=\int_\mathbb{R}e^{it\lambda}\,dv(\lambda),
\]
where $v$ is a non-decreasing right-continuous bounded function.
If $f(0)=1$ then such a function $v$ defines a probability measure
on $\mathbb{R}$ and Bochner's theorem says that $f$ is the Fourier
transform  of a probability measure, i.e., the characteristic
function of a random variable that follows the probability
distribution defined by $v$. Moreover, the condition of positive
definiteness of $f$ is necessary and sufficient for such a
representation.  The function  $ f(t)=\langle u, e^{it H}u\rangle
$, where $u$ is a unit vector and $H$ is a self-adjoint operator
 as described above, is an example of such a positive
definite function.

In this paper we use this spectral theorem based approach to
compute the characteristic function of several quantum random
variables $H$ defined as self-adjoint sums of the generators of
some  finite dimensional Lie algebras of interest in quantum
mechanics (the only  reason why a Lie structure is assumed is
because splitting the exponential of a sum of operators is usually
done through a Campbell-Baker-Hausdorff type formula that relies
on commutation relations).

 If $\mathcal{H}=\mathbb{R}^n$ or
$\mathcal{H}=\mathbb{C}^n$ for vectors $u=(u_1,...,u_n)$ and
$v=(v_1,...,v_n)$ in $\mathcal{H}$  we will use the standard inner
products $\langle u, v \rangle=uv^T$ and $\langle u, v
\rangle=\bar{u}v^T$ respectively. The  identity matrix/operator is
denoted by $I$, while $\delta$ denotes Dirac's delta function
defined, for a \textit{test function} $\phi$, by
\[
\int_{\mathbb{R}} \delta (x-a) \phi (x)\,dx= \int_{\mathbb{R}}
\delta (a-x) \phi (x)\,dx=\phi (a).
\]
 We define the Fourier transform of $f$ by
\[
\hat{f}(t)=\left(Uf\right)(t)=(2
\pi)^{-1/2}\,\int_{-\infty}^\infty \,e^{i\lambda
t}f(\lambda)d\lambda ,
\]
and the inverse Fourier transform of $\hat{f}$ by
\[
f(\lambda)=\left(U^{-1}\hat{f}\right)(\lambda)=(2
\pi)^{-1/2}\,\int_{-\infty}^\infty \,e^{-i\lambda t}\hat{f}(t)dt,
\]
so the inverse Fourier transform of
\[
(2 \pi)^{-1/2}\,\langle u, e^{itH}u \rangle
\]
gives the probability density function $p(\lambda)$ of $H$.

For $a\in\mathbb{C}$, let
 \[
 R(a;H)=(a-H)^{-1}
 \]
  denote the
resolvent of an operator $H$. The spectral resolution $\{E_\lambda
\,|\,\lambda \in \mathbb{R} \}$ of a bounded or unbounded
self-adjoint operator $H$ in a complex separable Hilbert Space
$\mathcal{H}$, is given by \textit{Stone's formula}  (see
\cite{DS}, Theorems X.6.1 and XII.2.10))
\[
E\left((a,b)\right)=\lim_{\delta\to 0^+}\lim_{\epsilon\to
0^+}\frac{1}{2\pi i}\int_{a+\delta}^{b-\delta}\left(R(t-\epsilon
i; H)-R(t+\epsilon i; H)\right)\,dt,
\]
where $(a,b)$ is the open interval $a<\lambda<b$, $R(t\pm\epsilon
i; H)=\left(t\pm\epsilon i-H\right)^{-1}$, and the limit is in the
strong operator topology. For $a\to-\infty$ and $b=\lambda$ we
have
\begin{align*}
E_\lambda&=E\left( (-\infty, \lambda ]\right)=\lim_{\rho\to
0^+}E\left( (-\infty,
\lambda+\rho)\right)\\
&=\lim_{\rho\to 0^+}\lim_{\delta\to 0^+}\lim_{\epsilon\to
0^+}\frac{1}{2\pi
i}\int_{-\infty}^{\lambda+\rho-\delta}\left(R(t-\epsilon i;
H)-R(t+\epsilon i; H)\right)\,dt.
\end{align*}
In particular (see \cite{Ro}, Theorem 4.31), for
$f,g\in\mathcal{H}$,
\[
\langle f, E_\lambda g\rangle=\lim_{\rho\to 0^+}\lim_{\delta\to
0^+}\lim_{\epsilon\to 0^+}\frac{1}{2\pi
i}\int_{-\infty}^{\lambda+\rho-\delta}\langle f,
\left(R(t-\epsilon i; H)-R(t+\epsilon i; H)\right)g\rangle\,dt.
\]
Thus, for a unit vector $u$, the \textit{vacuum resolution of the
identity} (terminology coming from the case when $u=\Phi$, the
\textit{vacuum vector} in a Fock space) of the operator $H$ is
given by
\begin{equation}\label{srv}
\langle u, E_\lambda u\rangle=\lim_{\rho\to 0^+}\lim_{\delta\to
0^+}\lim_{\epsilon\to 0^+}\frac{1}{2\pi
i}\int_{-\infty}^{\lambda+\rho-\delta}\langle u,
\left(R(t-\epsilon i; H)-R(t+\epsilon i; H)\right)u\rangle\,dt.
\end{equation}
 In  Section \ref{Stone} we will
show how, using formula (\ref{srv}), we can avoid the reliance on
\textit{splitting or disentanglement lemmas}, such as Lemma
\ref{3} of Section \ref{bo} for the splitting of operator
exponentials, in order to compute the characteristic function of a
quantum random variable. In particular, Stone's formula frees us
from any dependence on Lie algebraic structures. However, the
difficulty of obtaining a splitting lemma, is replaced by that of
computing the resolvent and the resulting spectral integrals.

\section{Quantum Observables  in $\mathfrak{sl}(2,\mathbb{R})$}
The Lie algebra $\mathfrak{sl}(2, \mathbb{R})$ of real $(2\times
2)$ matrices of zero trace, is generated \cite{Fein} by the
matrices
\[
\Delta=
\begin{pmatrix}
0&0\\
-1&0
\end{pmatrix}
 ,  R=
\begin{pmatrix}
0&1\\
0&0
\end{pmatrix}
 ,  \rho=
\begin{pmatrix}
1&0\\
0&-1
\end{pmatrix},
\]
with commutation relations
\[
\lbrack \Delta, R\rbrack=\rho,  \lbrack\rho, R\rbrack=2R, \lbrack
\rho , \Delta\rbrack=-2\Delta .
\]
We notice that the matrix
\[
H=R-\Delta+\rho=
\begin{pmatrix}
 1 & 1\\
1 &-1
\end{pmatrix}
\]
is real symmetric, thus it is a quantum observable. Its
eigenvalues are
\[
\lambda_1=-\sqrt{2},  \lambda_2=\sqrt{2},
\]
with corresponding eigenspaces
\begin{align*}
V_1&=\{xv_1: v_1=
\begin{pmatrix}
 1 -\sqrt{2}\\1
\end{pmatrix}
  ,  x\in\mathbb{R} \},
\\
V_2&=\{xv_1: v_2=
\begin{pmatrix}
 1 +\sqrt{2}\\1
\end{pmatrix}
  ,  x\in\mathbb{R} \},
\end{align*}
corresponding normalized basic eigenvectors
\[
u_1=\frac{v_1}{||v_1||}=
\begin{pmatrix}
 \frac{1 -\sqrt{2}}{\sqrt{4 -2\sqrt{2}}}\\
 \\
\frac{1}{\sqrt{4 -2\sqrt{2}}}
\end{pmatrix}
   ,  u_2=\frac{v_2}{||v_2||}=
\begin{pmatrix}
 \frac{1 +\sqrt{2}}{\sqrt{4 +2\sqrt{2}}}\\
 \\
\frac{1}{\sqrt{4 +2\sqrt{2}}}
\end{pmatrix},
\]
and eigen-projections
\begin{align*}
E_1&=\langle u_1, u_1\rangle=u_1^Tu_1=   \begin{pmatrix}
 \frac{1}{4}(2 -\sqrt{2})&-\frac{1}{2\sqrt{2}}  \\
 \\
-\frac{1}{2\sqrt{2}} &\frac{1}{4 -2\sqrt{2}}
\end{pmatrix},
  \\
&\\
 E_2&=\langle u_2, u_2\rangle=u_2^Tu_2=   \begin{pmatrix}
 \frac{1}{4}(2 +\sqrt{2})&\frac{1}{2\sqrt{2}}  \\
 \\
\frac{1}{2\sqrt{2}} &\frac{1}{4 +2\sqrt{2}}
\end{pmatrix}.
\end{align*}
We notice that $E_1$ and $E_2$ are a resolution of the identity,
i.e.,
\[
I=E_1+E_2
\]
and
\[
 H=\lambda_1E_1+\lambda_2E_2.
\]
Moreover
\[
e^{itH}=   \begin{pmatrix}
 \cos(\sqrt{2}t)+i\frac{1}{\sqrt{2}}\sin(\sqrt{2}t) & i\frac{1}{\sqrt{2}}\sin(\sqrt{2}t)\\
 \\
i\frac{1}{\sqrt{2}}\sin(\sqrt{2}t)
&\cos(\sqrt{2}t)-i\frac{1}{\sqrt{2}}\sin(\sqrt{2}t)
\end{pmatrix}.
\]
 If $u=(a,b)$ is a unit vector in $\mathbb{R}^2$ then for
$t\in\mathbb{R}$,
\begin{align*}
\langle u, e^{itH}u \rangle&=e^{it\lambda_1}\langle u, E_1u
\rangle+e^{it\lambda_2}\langle u, E_2 u \rangle=e^{it\lambda_1}
u^T E_1u +e^{it\lambda_2} u^T E_2 u\\
&=(a^2+b^2)\cos(\sqrt{2}t)+i\frac{a^2-b^2+2ab
}{\sqrt{2}}\sin(\sqrt{2}t) \\
&=\cos(\sqrt{2}t)+i\frac{a^2-b^2+2ab }{\sqrt{2}}\sin(\sqrt{2}t),
\end{align*}
so $H$ follows a Bernoulli distribution with
probability density function
\begin{align*}
p_{a,b}(\lambda)&=\frac{1}{4}\left( \left(2 + \sqrt{2} (a^2 + 2 a
b- b^2)\right) \delta(\sqrt{2} - \lambda)\right.\\
&\left. + \left(2 - \sqrt{2} (a^2 + 2 a b - b^2)\right)
\delta(\sqrt{2} + \lambda)\right),
\end{align*}
that is, $H$ takes the values  $\lambda_1=-\sqrt{2}$ and
 $\lambda_2=\sqrt{2}$ with probabilities
 \[
 P(H=-\sqrt{2})=\frac{1}{4}\left(2 - \sqrt{2} (a^2 + 2 a b - b^2)\right)
 \]
 and
 \[
  P(H=\sqrt{2})=\frac{1}{4}\left(2 + \sqrt{2} (a^2 + 2 a b - b^2)\right)
 \]
  respectively.  In particular if $u=(a,b)$ is a
\textit{Fock vacuum vector}, i.e. , if we require \cite{Fein} that
\[
\Delta u=\mathbf{0} \mbox{    and    }\rho u=c u
\,,\,c\in\mathbb{R},
\]
then we find that $c=-1$ and $a=0$ therefore $b=\pm 1$ and we
obtain the characteristic function in the \textit{vacuum state}
$\Phi= (0,\pm 1)$,
\[
\langle \Phi, e^{itH}\Phi \rangle=\cos(\sqrt{2}t)-i\frac{1
}{\sqrt{2}}\sin(\sqrt{2}t),
\]
so $H$ follows a Bernoulli distribution with probability density
function
\[
p_{0,\pm 1}(\lambda)=\frac{1}{4}\left( \left(2 - \sqrt{2}\right)
\delta(\sqrt{2} - \lambda) +
 \left(2 + \sqrt{2} \right) \delta(\sqrt{2} +
 \lambda)\right).
\]
The matrix
\[
H_0=R-\Delta=
\begin{pmatrix}
 0 & 1\\
1 &0
\end{pmatrix}
\]
is also an observable with spectral resolution
\[
H_0=\lambda_1E_1+\lambda_2E_2=(-1)\cdot
\begin{pmatrix}
 \frac{1}{2} &  -\frac{1}{2}\\
 -\frac{1}{2} & \frac{1}{2}
\end{pmatrix}
  +1\cdot
\begin{pmatrix}
 \frac{1}{2} &  \frac{1}{2}\\
 \frac{1}{2} & \frac{1}{2}
\end{pmatrix}
\]
and
\[
e^{itH_0}=
\begin{pmatrix}
\cos t & i \sin t\\
 i \sin t & \cos t
\end{pmatrix}.
\]
The characteristic function of $H_0$ is
\[
\langle u, e^{itH_0}u \rangle=(a^2+b^2)\cos t+2 i ab \sin t =\cos
t+2 i ab \sin t,
\]
so $H$ follows a Bernoulli distribution with probability density
function
\[
p_{a,b}(\lambda)= \left(\frac{1}{2} +ab \right) \delta( \lambda-1)
+\left(\frac{1}{2} -ab \right) \delta( \lambda+1).
\]
In the Fock vacuum state $\Phi=(0,\pm 1)$ the characteristic
function reduces (see also \cite{Fein}) to
\[
\langle \Phi, e^{itH_0}\Phi \rangle=\cos t,
\]
so
\[
p_{0,\pm 1}(\lambda)= \frac{1}{2}  \delta( \lambda-1) +\frac{1}{2}
 \delta( \lambda+1),
\]
while in the states
\[
u=\pm \frac{1}{\sqrt{2}}(1,1)
\]
we have
\[
\langle u, e^{itH_0}u \rangle=e^{it},
\]
so $H$ follows a discrete probability distribution with
probability density function
\[
p_{\pm \frac{1}{\sqrt{2}}(1,1)}(\lambda)=\delta( \lambda-1),
\]
i.e., in the state  defined by $u=\pm \frac{1}{\sqrt{2}}(1,1)$,
$H$ takes the value $\lambda=1$ with probability $1$.
 We remark that $H_0$ can also be regarded as a
\textit{Krawtchouk-Griffiths observable}  (see \cite{Fein2},
Example 5.12).
\section{Pauli Matrices and $\mathfrak{su}(2)$}\label{su}
The \textit{Pauli matrices} $\sigma_j$, $j=1,2,3$, of quantum
mechanics,
\[
\sigma_1=
\begin{pmatrix}
 0 & 1\\
1 &0
\end{pmatrix}
  ,  \sigma_2=
\begin{pmatrix}
 0 & -i\\
i &0
\end{pmatrix}
,  \sigma_3=
\begin{pmatrix}
 1 & 0\\
0 &-1
\end{pmatrix},
\]
with commutation relations
\[
\lbrack \sigma_1, \sigma_2\rbrack=2i\sigma_3,  \lbrack \sigma_2,
\sigma_3\rbrack=2i\sigma_1 ,  \lbrack \sigma_3,
\sigma_1\rbrack=2i\sigma_2,
\]
are Hermitian, i.e., self-adjoint. The matrices $i\sigma_1,
-i\sigma_2, i\sigma_3$ generate the Lie algebra
\[
\mathfrak{su}(2)=\{
\begin{pmatrix}
 ia & -\bar{z}\\
z &-ia
\end{pmatrix}  \,:\,a\in\mathbb{R}, z\in\mathbb{C}\}
\]
of traceless anti-hermitian $(2\times 2)$ matrices. The spectral
decompositions  of the  quantum observables corresponding to the
Pauli matrices are
\begin{align*}
\sigma_1&=(-1)\cdot
\begin{pmatrix}
 \frac{1}{2} & -\frac{1}{2}\\
-\frac{1}{2} &\frac{1}{2}
\end{pmatrix}
  +1\cdot
\begin{pmatrix}
 \frac{1}{2} & \frac{1}{2}\\
\frac{1}{2} &\frac{1}{2}
\end{pmatrix},
   \\
   \sigma_2&=(-1)\cdot
\begin{pmatrix}
 \frac{1}{2} & \frac{i}{2}\\
-\frac{i}{2} &\frac{1}{2}
\end{pmatrix}
  +1\cdot
\begin{pmatrix}
 \frac{1}{2} & -\frac{i}{2}\\
\frac{i}{2} &\frac{1}{2}
\end{pmatrix},
  \\
\sigma_3&=(-1)\cdot
\begin{pmatrix}
 0 & 0\\
0 & 1
\end{pmatrix}
  +1\cdot
\begin{pmatrix}
 1 & 0\\
0 &0
\end{pmatrix},
\end{align*}
with complex exponentials
\begin{align*}
e^{i t \sigma_1}&=
\begin{pmatrix}
 \cosh (it) & \sinh(it)\\
 \sinh(it)& \cosh (it)
\end{pmatrix}  =
\begin{pmatrix}
 \cos t & i\sin t\\
 i\sin t& \cos t
\end{pmatrix},
  \\
 e^{i t \sigma_2}&=
\begin{pmatrix}
 \cosh (it) & -i\sinh(it)\\
 i\sinh(it)& \cosh (it)
\end{pmatrix}  =
\begin{pmatrix}
 \cos t & \sin t\\
 -\sin t& \cos t
\end{pmatrix},
  \\
e^{i t \sigma_3}&=
\begin{pmatrix}
 e^{it} & 0\\
 0& e^{-it}
\end{pmatrix},
 \end{align*}
 while, for $u=(a,b)\in\mathbb{C}^2$ with $|a|^2+|b|^2=1$, using for
 $j=1,2,3$,
\[
\langle u, e^{it\sigma_j}u \rangle=e^{it\lambda_1}\langle u, E_1u
\rangle+e^{it\lambda_2}\langle u, E_2 u \rangle=e^{it\lambda_1}
\bar{u}^T E_1u +e^{it\lambda_2} \bar{u}^T E_2 u,
 \]
we obtain the   characteristic functions of $\sigma_1, \sigma_2,
\sigma_3$,
\begin{align*}
\langle u, e^{it\sigma_1}u \rangle& =\cosh
(it)+(b\bar{a}+a\bar{b})\sinh (it)=\cos t+i(b\bar{a}+a\bar{b})\sin
t\\
\langle u, e^{it\sigma_2}u \rangle&= \cosh
(it)+i(a\bar{b}-b\bar{a})\sinh (it)= \cos
t+(b\bar{a}-a\bar{b})\sin t \\
\langle u, e^{it\sigma_3}u \rangle&= |a|^2e^{it}+|b|^2e^{-it},
\end{align*}
so $\sigma_1, \sigma_2, \sigma_3$ are quantum random variables
following a Bernoulli distribution with probability density
function
\begin{align*}
p_1(\lambda)&=  \frac{1}{2}\left(1 + b \bar{a} + a \bar{b}\right)
\delta(\lambda-1) +\frac{1}{2}\left(
 1 - b \bar{a} - a \bar{b}\right) \delta(\lambda+1),\\
p_2(\lambda)&=\frac{1}{2}\left(1 + b \bar{a} - a \bar{b}\right)
\delta(\lambda-1) +\frac{1}{2}\left( 1 - b \bar{a} + a \bar{b}\right) \delta(\lambda+1),   \\
p_3(\lambda)&= |a|^2 \delta(\lambda-1) +|b|^2 \delta(\lambda+1),
\end{align*}
respectively.
\section{Pauli Matrices and $\mathfrak{su}(1,1)$}\label{su11}
The matrices
\[
K_1=\frac{i}{2} \sigma_2=
\begin{pmatrix}
 0 & \frac{1}{2}\\
-\frac{1}{2} &0
\end{pmatrix}
, K_2=-\frac{i}{2}\ \sigma_1=
\begin{pmatrix}
 0 & -\frac{i}{2}\\
 -\frac{i}{2}&0
\end{pmatrix}
  \]
and
 \[
 K_0=\frac{1}{2} \sigma_3=
\begin{pmatrix}
 \frac{1}{2} & 0\\
0 &-\frac{1}{2}
\end{pmatrix},
\]
satisfy (see \cite{Brazil}) the $\mathfrak{su}(1,1)$ Lie algebra
 commutation relations:
\[
\lbrack K_1, K_2\rbrack=-i K_0 ,  \lbrack K_0, K_1\rbrack=i K_2,
\lbrack K_2 , K_0\rbrack=i K_1.
\]
 The matrix
\[
H=i(K_1+K_2)+K_0=\frac{1}{2}\begin{pmatrix}
 1 & 1+i\\
1-i &-1
\end{pmatrix}
\]
is Hermitian so it is a quantum observable. Its eigenvalues are
\[
\lambda_1=-\frac{\sqrt{3}}{2}, \lambda_2=\frac{\sqrt{3}}{2},
\]
with corresponding  normalized basic eigenvectors
\[
u_1= \frac{1}{\sqrt{3 -\sqrt{3}}}
\begin{pmatrix}
\frac{1-\sqrt{3}}{2}+i\frac{1-\sqrt{3}}{2} \\
\\
1
\end{pmatrix}
,  u_2=\frac{1}{\sqrt{3 +\sqrt{3}}}
\begin{pmatrix}
\frac{1+\sqrt{3}}{2}+i\frac{1+\sqrt{3}}{2} \\
\\
1
\end{pmatrix},
\]
and eigen-projections
\begin{align*}
E_1&=u_1^Tu_1=   \frac{1}{2}\begin{pmatrix}
 \frac{3 -\sqrt{3}}{3}&\frac{\sqrt{3}-1}{\sqrt{3}-3}(1+i)  \\
 \\
\frac{\sqrt{3}-1}{\sqrt{3}-3}(1-i)   &\frac{3 +\sqrt{3}}{3}
\end{pmatrix},
  \\
&\\
 E_2&=u_2^Tu_2=  \frac{1}{2}\begin{pmatrix}
 \frac{3 +\sqrt{3}}{3}&\frac{\sqrt{3}+1}{\sqrt{3}+3}(1+i)  \\
 \\
\frac{\sqrt{3}+1}{\sqrt{3}+3}(1-i)   &\frac{3 -\sqrt{3}}{3}
\end{pmatrix}.
\end{align*}
We have
\[
I=E_1+E_2
\]
and
\[
 H=\lambda_1E_1+\lambda_2E_2.
\]
Moreover
\[
e^{itH}=   \begin{pmatrix}
\cos \left(\frac{\sqrt{3}t}{2} \right)+i\frac{1}{\sqrt{3}}\sin \left(\frac{\sqrt{3}t}{2} \right) & \frac{i-1}{\sqrt{3}}\sin \left(\frac{\sqrt{3}t}{2} \right)\\
 \\
\frac{i+1}{\sqrt{3}}\sin \left(\frac{\sqrt{3}t}{2} \right) &\cos
\left(\frac{\sqrt{3}t}{2} \right)-i\frac{1}{\sqrt{3}}\sin
\left(\frac{\sqrt{3}t}{2} \right)
\end{pmatrix}.
\]
 If $u=(a,b)$ is a unit vector in $\mathbb{C}^2$ then for
$t\in\mathbb{R}$
\[
\langle u, e^{itH}u \rangle=e^{it\lambda_1}\langle u, E_1u
\rangle+e^{it\lambda_2}\langle u, E_2 u \rangle=e^{it\lambda_1}
\bar{u}^T E_1u +e^{it\lambda_2} \bar{u}^T E_2 u
\]
\[
= \left(\frac{i-1}{\sqrt{3}}\,b\bar{a}+ \frac{i+1}{\sqrt{3}} \, a
\bar{b}+ \frac{i}{\sqrt{3}}(|a|^2-|b|^2)\right) \sin
\left(\frac{\sqrt{3}t}{2} \right) +\cos \left(\frac{\sqrt{3}t}{2}
\right),
\]
so $H$ follows a Bernoulli distribution with probability density
function
\begin{align*}
p_{a,b}(\lambda)&=\frac{1}{2}\left(1 + \frac{1}{\sqrt{
    3}}\left((a + (1 + i) b) \bar{a} +((1 - i) a - b)
    \bar{b}\right)\right)
    \delta\left( \frac{\sqrt{3}}{2} -  \lambda\right) \\
&+ \frac{1}{2}\left(1 - \frac{1}{\sqrt{3}}\left((a + (1 + i)
b)\bar{a} + ((1 - i) a -
b)\bar{b}\right)\right)\delta\left(\frac{\sqrt{3}}{2} +
\lambda\right).
\end{align*}
In particular, for $a=1$ and $b=0$,
\[
\langle u, e^{itH}u \rangle=\cos \left(\frac{\sqrt{3}t}{2}
\right)+ \frac{i}{\sqrt{3}} \sin \left(\frac{\sqrt{3}t}{2}
\right),
\]
i.e., $H$ follows a Bernoulli distribution with probability
density function
\begin{align*}
p_{1,0}(\lambda)&=\frac{1}{2}\left(1 + \frac{1}{\sqrt{
    3}} \right)
    \delta\left(\frac{\sqrt{3}}{2} -  \lambda\right)
     + \frac{1}{2}\left(1 - \frac{1}{\sqrt{3}}\right)
     \delta\left(\frac{\sqrt{3}}{2} +  \lambda\right),
\end{align*}
while for $a=0$ and $b=1$,
\[
\langle u, e^{itH}u \rangle=\cos \left(\frac{\sqrt{3}t}{2}
\right)- \frac{i}{\sqrt{3}} \sin \left(\frac{\sqrt{3}t}{2}
\right),
\]
i.e., $H$ follows a Bernoulli distribution with probability
density function
\begin{align*}
p_{0,1}(\lambda)&=\frac{1}{2}\left(1 - \frac{1}{\sqrt{
    3}} \right)
    \delta\left(\frac{\sqrt{3}}{2} -  \lambda\right)
     + \frac{1}{2}\left(1 + \frac{1}{\sqrt{3}}\right)
     \delta\left(\frac{\sqrt{3}}{2} +  \lambda\right).
\end{align*}
\section{The Casimir Element  in $\mathfrak{so}(3)$}
The Lie algebra $\mathfrak{so}(3)$ of $(3\times 3)$ skew-symmetric
matrices is generated by the matrices
\[
L_x=
\begin{pmatrix}
  0& 0&  0\\
 0 &0&-1 \\
 0& 1& 0
\end{pmatrix}
   ,  L_y=
\begin{pmatrix}
  0& 0&  1\\
 0 &0&0 \\
 -1& 0& 0
\end{pmatrix}
   ,  L_z=
\begin{pmatrix}
  0& -1&  0\\
 1 &0&0 \\
 0& 0& 0
\end{pmatrix},
\]
with commutation relations
\[
\lbrack L_x, L_y\rbrack=L_z,  \lbrack L_y, L_z\rbrack=L_x, \lbrack
L_z, L_x\rbrack=L_y.
\]
Associated with $L_x, L_y, L_z$ is the self-adjoint central
\textit{Casimir element}
\[
L=L_x^2+L_y^2+L_z^2=-2 I_3 ,
\]
where $I_3$ is the $(3\times 3)$ identity matrix.
 For a unit vector  $u=(a,b,c)$ in $\mathbb{R}^3$ we have
\[
\langle u, e^{itL}u \rangle=e^{-2it},
\]
so $L$ follows a discrete probability distribution with
probability density function
\[
p(\lambda)=\delta( \lambda+2),
\]
i.e., in the state defined by $u$, $H$ takes the value
$\lambda=-2$ with probability $1$.
\section{Quantum Observables  in $\mathfrak{h}(3,\mathbb{R})$}
The \textit{Heisenberg algebra} $\mathfrak{h}$ is the
three-dimensional Lie algebra with generators $D, X, h$ satisfying
the commutation relations
\[
\lbrack D, X\rbrack=h ,  \lbrack D, h \rbrack=\lbrack X, h
\rbrack=0.
\]
A matrix representation of $\mathfrak{h}$ is provided by the
$3$-dimensional matrix Lie algebra $\mathfrak{h}(3,\mathbb{R})$
defined as the vector space of matrices of the form
\[
A=
\begin{pmatrix}
  0& x&  y\\
 0 &0&z \\
 0& 0& 0
\end{pmatrix}  ,  x, y, z \in
\mathbb{R},
\]
spanned by the matrices
\[
D=
\begin{pmatrix}
  0& 1&  0\\
 0 &0&0 \\
 0& 0& 0
\end{pmatrix}
   ,   X=
\begin{pmatrix}
  0& 0&  0\\
 0 &0&1 \\
 0& 0& 0
\end{pmatrix}
, h=\begin{pmatrix}
  0& 0&  1\\
 0 &0&0 \\
 0& 0& 0
\end{pmatrix},
\]
which are easily seen to satisfy the commutation relations
\[
\lbrack D, X\rbrack=h ,  \lbrack D, h \rbrack=\lbrack X, h
\rbrack=0.
\]
Clearly, no linear combination of $D, X$ and $h$ can be symmetric
on all of $\mathbb{R}^3$ with respect to the usual inner product.
Nevertheless, the matrix
\[
H=(D+X+h)+(D+X+h)^T=\begin{pmatrix}0&1&1\\
1&0&1\\
1&1&0
\end{pmatrix}=J-I,
\]
where $J$ is the all-ones $(3\times 3)$ matrix, is a quantum
observable with eigenvalues:  $\lambda_1=2$, with multiplicity one
and corresponding normalized basic eigenvector
\[
u_1=\frac{1}{\sqrt{3}} (1, 1, 1),
\]
and $\lambda_2=-1$ with multiplicity $2$ and corresponding
orthonormalized basis eigenvectors
\[
u_2=\frac{1}{\sqrt{2}} ( -1, 0, 1)  , u_3=\frac{1}{\sqrt{6}}  (-1,
2, -1).
\]
The associated eigen-projections are
\begin{align*}
E_1&=\langle u_1,
u_1\rangle=u_1^Tu_1=\frac{1}{3}\begin{pmatrix}1&1&1\\
1&1&1\\
1&1&1
\end{pmatrix}\\
E_2&=\langle u_2, u_2\rangle+\langle u_3,
u_3\rangle  =u_2^Tu_2+u_3^Tu_3\\
&=\frac{1}{2}\begin{pmatrix}1&0&-1\\
0&0&0\\
-1&0&1
\end{pmatrix}+\frac{1}{6}\begin{pmatrix}1&-2&1\\
-2&4&-2\\
1&-2&1
\end{pmatrix}\\
&=\frac{1}{3}\begin{pmatrix}2&-1&-1\\
-1&2&-1\\
-1&-1&2
\end{pmatrix},
\end{align*}
with
\[
I=E_1+E_2,
\]
and
\[
H=\lambda_1E_1+\lambda_2E_2.
\]
If $u=(a,b, c)$ is a unit vector in $\mathbb{R}^3$ then for
$t\in\mathbb{R}$:
\begin{align*}
\langle u, e^{itH}u \rangle&=e^{it\lambda_1}\langle u, E_1u
\rangle+e^{it\lambda_2}\langle u, E_2 u
\rangle \\
&=\left(1-\frac{(a+b+c)^2}{3}\right) e^{-i
t}+\frac{(a+b+c)^2}{3}e^{2 i t},
\end{align*}
i.e., $H$ follows a Bernoulli distribution with probability
density function
\begin{align*}
p_{a,b,c}(\lambda)&=\left(1-\frac{(a+b+c)^2}{3}\right) \delta(
\lambda+1)+\frac{(a+b+c)^2}{3}\delta( \lambda-2).
\end{align*}
\section{Multiplication and Differentiation
$L^2$-Operators}\label{Heis}
 The classical Heisenberg Lie algebra of quantum mechanics  has
generators  $P, X, I$ and non-zero commutation relations among
generators
\[
\lbrack P, X \rbrack=-i\hbar I.
\]
 The (self-adjoint) position, momentum and identity
operators defined in $L^2(\mathbb{R},\mathbb{C})$ with inner
product
\[
\langle f, g \rangle
=\frac{1}{\sqrt{\hbar}}\,\int_{\mathbb{R}}\,\overline{f(x)}g(x)\,dx
,
\]
by
\[
(X\,f)(x)=x\,f(x), (P\,f)(x)=-i\,\hbar\,f^{\prime}(x), (I
f)(x)=f(x),
\]
realize the Heisenberg commutation relations on ${\rm dom }(X)\cap
{\rm dom }(P)$, where
\[
{\rm dom }(X)=\{f\in L^2(\mathbb{R},\mathbb{C}))
\,:\,\int_{\mathbb{R}}x^2\,|f(x)|^2\,dx<+\infty\}
\]
and
\[
{\rm dom }(P)=\{f\in L^2(\mathbb{R},\mathbb{C}) \,:\,\mbox{ $f$ is
absolutely continuous,
}\int_{\mathbb{R}}\,\left|\frac{df(x)}{dx}\right|^2\,dx<
+\infty\},
\]
are respectively the, dense in $L^2(\mathbb{R},\mathbb{C})$,
domains of $X$ and $P$ (see \cite{QMI}, Sec. 2.3 and
\cite{Yosida}, Sec VII 3). Functions in the domain of $P$ are
continuous and vanish at infinity (see \cite{Richtmyer}, Section
5.6). Using
\begin{equation}\label{pmwn}
X=\sqrt{\hbar}\,\frac{a+a^{\dagger}}{\sqrt{2}}
\,\,,\,\,P=\sqrt{\hbar}\,\frac{a-a^{\dagger}}{\sqrt{2}i},
\end{equation}
we obtain the \textit{Boson pair}
\begin{equation}\label{wnpm}
a=\frac{X+iP}{\sqrt{2\hbar}}
\,\,,\,\,a^{\dagger}=\frac{X-iP}{\sqrt{2\hbar}},
\end{equation}
with
\[
\lbrack a, a^{\dagger}\rbrack=\mathbf{1}\,\,,\,\,
a^*=a^{\dagger}\,\,,\,\,a\Phi=0\,\,,\mbox{where }\,\,\Phi=\Phi(x)=
\pi^{-1/4}\,e^{-\frac{x^2}{2\hbar}}.
\]
We notice that
\[
\|\Phi\|^2=\langle \Phi, \Phi \rangle =1 .
\]
Moreover
\[
\Phi^{\prime}(x)=-\frac{1}{\hbar \pi^{1/4}} x
e^{-\frac{x^2}{2\hbar}}\,\,;\,\,\lim_{x\to\pm \infty}\Phi(x)=0,
\]
so $\Phi^{\prime}$ is bounded on $\mathbb{R}$, which means that
$\Phi$ is Lipschitz and therefore absolutely continuous on
$\mathbb{R}$. Since
\[
\int_{\mathbb{R}}x^2\,|\Phi(x)|^2\,dx= \frac{\hbar^{3/2}}{2}
<+\infty, \int_{\mathbb{R}}\left|\frac{d\Phi(x)}{dx}\right|^2\,dx=
\frac{1}{\hbar^{1/2}}  < +\infty ,
\]
 $\Phi$ will be our prototype unit vector in ${\rm dom }(X)\cap
{\rm dom }(P)$.
\begin{theorem}\label{1Dh} If units are chosen so that $\hbar=1$ then
 the vacuum characteristic function of the quantum observable
\[
H=X+P
\]
 is
\[
\langle \Phi, e^{itH} \Phi \rangle=e^{-\frac{t^2}{2}},
\]
i.e., the underlying probability distribution is Gaussian.
\end{theorem}
\begin{proof} The spectral resolutions of $X$ and $P$ are (see
\cite{Yosida}, Sec. XI.5 and XI.6)
\[
X=\int_{\mathbb{R}}\lambda \,dE_\lambda
\,\,;\,\,P=\int_{\mathbb{R}}\lambda
\,dE^{\prime}_\lambda=\int_{\mathbb{R}}\lambda \,d(UE_\lambda
U^{-1}),
\]
where
\[
E_\lambda f(t)=\left\{
\begin{array}{llr}
 f(t)   &\mbox{ if } t\leq \lambda   \\
0&\mbox{ if } t>\lambda
\end{array}
\right.
\]
and
\[
E^{\prime}_\lambda=U\,E_\lambda\,U^{-1},
\]
where
\[
\left(Uf\right)(t)=(2 \pi)^{-1/2}\,\int_{-\infty}^\infty
\,e^{ist}f(s)ds =\lim_{n\to +\infty}\,(2 \pi)^{-1/2}\,\int_{-n}^n
\,e^{ist}f(s)ds
\]
is the Fourier transform of $f$ and
\[
\left(U^{-1} f\right)(t)=\left(U^{*} f\right)(t)=\left(U
f\right)(-t).
\]
By the Campbell--Baker--Hausdorff formula \cite{Hall} (see also
\cite{Fein}, splitting formula for the Heisenberg group) we have
\[
e^{itP}e^{itX}=e^{it(P+X)+\frac{1}{2}\lbrack itP, itX
\rbrack}=e^{it(P+X)-\frac{t^2}{2}(-i \hbar)}=e^{it(P+X)}e^{\frac{i
\hbar t^2}{2} },
\]
where both sides of the above are  unitary operators. Therefore,
\[
e^{it(P+X)}=e^{itP}e^{itX}e^{-\frac{i \hbar t^2}{2}
}=e^{itP}e^{itX}e^{-\frac{i  t^2}{2} }.
\]
Thus
\begin{align*}
\langle \Phi, e^{itH} \Phi \rangle&=e^{-\frac{i  t^2}{2} }\langle
\Phi, e^{itP}e^{itX} \Phi \rangle=e^{-\frac{i  t^2}{2} }\langle
e^{-itP}\Phi, e^{itX} \Phi \rangle   \\
&=e^{-\frac{i  t^2}{2} }\langle
\int_{\mathbb{R}}e^{-it\lambda^\prime} \,d(UE_{\lambda^\prime}
U^{-1})\Phi, \int_{\mathbb{R}}e^{it\lambda} \,dE_\lambda \Phi
\rangle\\
&=e^{-\frac{i  t^2}{2}
}\int_{\mathbb{R}}e^{it\lambda}\,d_\lambda\langle
\int_{\mathbb{R}}e^{-it\lambda^\prime} \,d(UE_{\lambda^\prime}
U^{-1})\Phi, E_\lambda \Phi \rangle\\
&=e^{-\frac{i  t^2}{2}
}\int_{\mathbb{R}}\int_{\mathbb{R}}e^{it\lambda}e^{it\lambda^\prime}\,d_\lambda\,d_{\lambda^\prime
}\langle E_{\lambda^\prime} U^{-1}\Phi, U^*E_\lambda \Phi \rangle
.
\end{align*}
Now,
\[
(U^{-1}\Phi)(t)=(U\Phi)(-t) =(2 \pi)^{-1/2}\,\int_{-\infty}^\infty
\,e^{-ist}\Phi(s)ds= {\pi}^{-1/4}e^{-\frac{t^2}{2}},
\]
so
\[
 (E_{\lambda^\prime}
 U^{-1}\Phi)(t)={\pi}^{-1/4}e^{-\frac{t^2}{2}}
 \chi_{ { _{( -\infty, \lambda^\prime\rbrack }}}(t).
\]
Similarly,
\[
(E_{\lambda}\Phi)(t)=\Phi(t)\chi_{ { _{( -\infty, \lambda\rbrack
}}}(t),
\]
where
\[
\chi_{ { _{( -\infty, \lambda\rbrack }}}(t)=\left\{
\begin{array}{llr}
1   &\mbox{ if } t\leq \lambda   \\
0&\mbox{ if } t>\lambda
\end{array}
\right.
\]
and
\begin{align*}
(U^{*}E_{\lambda}\Phi)(t)=(UE_{\lambda}\Phi)(-t) &=(2
\pi)^{-1/2}\,\int_{-\infty}^\infty \,e^{-ist}\Phi(s)\chi_{ { _{(
-\infty, \lambda\rbrack }}}(s)ds\\
&=(2 \pi)^{-1/2}\,\int_{-\infty}^\lambda \,e^{-ist}\Phi(s)ds.
\end{align*}
 Thus
\begin{align*}
d_\lambda\,d_{\lambda^\prime }\langle E_{\lambda^\prime}
U^{-1}\Phi, U^*E_\lambda \Phi
\rangle&=d_\lambda\,d_{\lambda^\prime }\int_{-\infty}^\infty\,\overline{(E_{\lambda^\prime}
 U^{-1}\Phi)(t)}\,(U^{*}E_{\lambda}\Phi)(t)\,dt
\\
&=\frac{1}{\pi\sqrt{2}}\,d_\lambda\,d_{\lambda^\prime
}\int_{-\infty}^\infty \,e^{-\frac{t^2}{2}}\chi_{ { _{( -\infty,
\lambda^\prime\rbrack }}}(t) \,\int_{-\infty}^\lambda \,e^{-ist}e^{-\frac{s^2}{2}}\,ds\,dt\\
&=\frac{1}{\pi\sqrt{2}}\,d_\lambda\,d_{\lambda^\prime
}\int_{-\infty}^{\lambda^\prime} \,e^{-\frac{t^2}{2}} \,\int_{-\infty}^\lambda \,e^{-ist}e^{-\frac{s^2}{2}}\,ds\,dt\\
&=\frac{1}{\pi\sqrt{2}}\,d_\lambda\left(
e^{-\frac{{\lambda^\prime}^2}{2}} \,\int_{-\infty}^\lambda
\,e^{-is\lambda^\prime}e^{-\frac{s^2}{2}}\,ds\,d\lambda^\prime\right)\\
&=\frac{1}{\pi\sqrt{2}}\,\left( e^{-\frac{{\lambda^\prime}^2}{2}}
 \,e^{-i\lambda\lambda^\prime}e^{-\frac{\lambda^2}{2}}\right)\,d\lambda\,d{\lambda^\prime}.
\end{align*}
Therefore, using  the integration formula
\[
\int_{-\infty}^\infty e^{-(a x+i
b)^2}\,dx=\frac{\sqrt{\pi}}{a}\,\,;\,\,a,b\in\mathbb{R},\,\,a>0,
\]
 twice, we obtain
\begin{align*}
\langle \Phi, e^{itH} \Phi \rangle&=e^{-\frac{it^2}{2}
}\int_{\mathbb{R}}\int_{\mathbb{R}}e^{it\lambda}e^{it\lambda^{\prime}}
\frac{1}{\pi\sqrt{2}}\,\left(e^{-\frac{{\lambda^{\prime}}^2}{2}}
 \,e^{-i\lambda\lambda^{\prime}}e^{-\frac{\lambda^2}{2}}\right)\,d_\lambda\,d_{\lambda^{\prime}}\\
&=\frac{1}{\pi\sqrt{2}}e^{-\frac{it^2}{2}}\int_{\mathbb{R}}e^{it\lambda^{\prime}
-\frac{{\lambda^{\prime}}^2}{2}}\,\left(\int_{\mathbb{R}}e^{it\lambda-\frac{\lambda^2}{2}-i
\lambda \lambda^\prime}\,d\lambda \right)\,d\lambda^{\prime}\\
&=\frac{1}{\pi\sqrt{2}}e^{-\frac{it^2}{2}}\int_{\mathbb{R}}e^{it\lambda^{\prime}
-\frac{{\lambda^{\prime}}^2}{2}}\,\left(e^{-2(\lambda^{\prime}-t)^2}
\int_{\mathbb{R}}e^{-\left(
\frac{1}{\sqrt{2}}\lambda-i\sqrt{2}(\lambda^\prime-t)\right)^2}\,d\lambda
\right)\,d\lambda^{\prime}\\
&=\frac{1}{\pi\sqrt{2}}e^{-\frac{it^2}{2}}\int_{\mathbb{R}}e^{it\lambda^{\prime}
-\frac{{\lambda^{\prime}}^2}{2}}\,\left(e^{-2(\lambda^{\prime}-t)^2}
\sqrt{2\pi} \right)\,d\lambda^{\prime}\\
&=\frac{1}{\sqrt{\pi}}e^{-\frac{(i+1)t^2}{2}}\int_{\mathbb{R}}
e^{-{\lambda^{\prime}}^2+(i+1)t\lambda^{\prime}
}d\lambda^{\prime}\\
&=\frac{1}{\sqrt{\pi}}e^{-\frac{(i+1)t^2}{2}}\left(e^{\frac{it^2}{2}}\int_{\mathbb{R}}
e^{-\left(\lambda^{\prime}-\frac{t(i+1)}{2}\right)^2}d\lambda^{\prime}\right)\\
&=\frac{1}{\sqrt{\pi}}e^{-\frac{(i+1)t^2}{2}}\left(e^{\frac{it^2}{2}}\sqrt{\pi}\right)\\
&=e^{-\frac{t^2}{2}}.
\end{align*}
\end{proof}
\section{Boson and $X,P$ Form of $\mathfrak{su}(1,1)$}\label{bo}
If $K_0, K_1, K_2$ satisfy the $\mathfrak{su}(1,1)$ commutation
relations then \cite{Brazil} the change of basis operators
\[
K_+=K_1+i K_2, K_-=K_1-i K_2, K_0,
\]
satisfy the commutation relations
\[
\lbrack K_0, K_\pm\rbrack=\pm K_\pm ,  \lbrack K_+, K_-\rbrack=-2
K_0,
\]
and have the boson realization
\[
K_+=\frac{1}{2}\left(a^{\dagger}\right)^2,
K_-=\frac{1}{2}\left(a\right)^2, K_0=\frac{1}{4}\left(a
a^{\dagger}+a^{\dagger}a\right).
\]
\begin{lemma}\label{1}
In terms of the position and momentum operators $X$ and $P$ of
Section \ref{Heis}, for $\hbar=1$,
\[
K_+ + K_- + K_0=\frac{1}{4}\left(3 X^2-P^2\right)+\frac{1}{2}I.
\]
\end{lemma}
\begin{proof} By (\ref{wnpm})
\begin{align*}
\left(a^{\dagger}\right)^2&=\frac{1}{2}\left(X- i P\right)\left(X- i P\right),\\
a^2&=\frac{1}{2}\left(X+ i P\right)\left(X+ i P\right),\\
a^{\dagger}a&=\frac{1}{2}\left(X+ i P\right)\left(X- i P\right).
\end{align*}
Multiplying out and using
\[
PX=XP-iI,
\]
we obtain
\begin{align*}
\left(a^{\dagger}\right)^2&=\frac{1}{2}\left(X^2-P^2-2 i X
P-I\right),\\
a^2&=\frac{1}{2}\left(X^2-P^2+2 i X
P+I\right),\\
a^{\dagger}a&=\frac{1}{2}\left(X^2+P^2+I\right).
\end{align*}
 Since
\[
a a^{\dagger}= I+a^{\dagger}a,
\]
we have
\begin{align*}
K_+ +K_- +K_0
&=\frac{1}{2}\left(a^{\dagger}\right)^2+\frac{1}{2}\left(a\right)^2
+\frac{1}{4}\left(a a^{\dagger}+a^{\dagger}a\right)\\
&=\frac{1}{4}\left(X^2-P^2-2 i X
P-I\right)+\frac{1}{4}\left(X^2-P^2+2 i X P+I\right)\\
&+\frac{1}{4}\left(\frac{1}{2}\left(X^2+P^2+I\right)+
I+\frac{1}{2}\left(X^2+P^2+I\right) \right)\\
&=\frac{1}{4}\left(3 X^2-P^2\right)+\frac{1}{2}I.
\end{align*}
\end{proof}
\begin{lemma}\label{2}
With $X$ and $P$ as in Section \ref{Heis} and $\hbar=1$,
\begin{align}
\lbrack e^{b X^2}, XP  \rbrack &=2 i b X^2 e^{b X^2},\label{21}\\
e^{i b X P}P^2&=e^{-2 b}P^2 e^{i b X P},\label{22}\\
\lbrack  e^{b X^2}, P^2 \rbrack &=\left( 2 b-4 b^2 X^2+4 i b X
P\right)e^{b X^2},\label{23}
\end{align}
for all $b\in \mathbb{C}$.
\end{lemma}
\begin{proof} The proof can be done  by working directly
with the $X, P$ commutation relations. However, (\ref{21}) and
(\ref{23}) can be obtained from formulas (vi) and (i),
respectively, of Lemma 3.2 of \cite{AccBouCOSA} with the
correspondence
\[
S^0_2=-P^2\,\,,\,\,S^2_0=X^2\,\,,\,\,S^1_1=\frac{1}{2}I+i X P.
\]
Formula (\ref{22}) is obtained, in like manner, from (2.11) of
Lemma 2 of \cite{AccBouArxiv}.
\end{proof}
\begin{lemma}\label{3}
With $X$ and $P$ as in Section \ref{Heis}, and $\hbar=1$, for all
$s\in\mathbb{C}$ and $ a, b \in\mathbb{R}$,
\begin{equation}\label{e1}
e^{s(a X^2+b P^2)}=e^{\frac{1}{2}p(s)}e^{q(s)X^2}e^{i p(s) X
P}e^{r(s)P^2},
\end{equation}
where
\begin{align}
q(s)&= \frac{1}{2}\sqrt{\frac{a}{b} }\tanh (2\sqrt{ab}\,s),   \label{e2}\\
p(s)&=-4 b \int_0^s q(t)\,dt=\log (\rm{sech}(2\sqrt{ab}\,s)),\label{e3}\\
r(s)&=b \int_0^s e^{2 p(t)}\,dt=b \int_0^s\left(\rm{sech}
(2\sqrt{ab}\,t)\right)^{2}\,dt=\frac{1}{2}\sqrt{\frac{b}{a} }\tanh
(2\sqrt{ab}\,s).  \label{e4}
\end{align}
\end{lemma}
\begin{proof} Let $R(s)$ and $L(s)$ be, respectively, the right
and left hand sides of equation (\ref{e1}). Then
\begin{equation}\label{L}
\frac{dL}{ds}=(a X^2+b P^2) L(s)\,\,,\,\,L(0)=I.
\end{equation}
We will show that
\begin{equation}\label{R}
\frac{dR}{ds}=(a X^2+b P^2) R(s)\,\,,\,\,R(0)=I,
\end{equation}
as well. That will imply  $L(s)=R(s)$ for all $s$. Clearly
$R(0)=I$. Moreover, direct differentiation gives
\begin{align*}
\frac{dR}{ds}&=\frac{1}{2}p^\prime
(s)R(s)+q^\prime(s)X^2R(s)+ip^\prime
(s)e^{\frac{1}{2}p(s)}e^{q(s)X^2}XP e^{i
p(s)XP}e^{r(s)P^2}\\
&+r^\prime(s)e^{\frac{1}{2}p(s)}e^{q(s)X^2}e^{i
p(s)XP}P^2e^{r(s)P^2}.
\end{align*}
By (\ref{21}) and (\ref{22}),
\[
e^{q(s)X^2}XP=\left(2 i q(s)X^2 +XP\right) e^{q(s)X^2},
 \]
 and
 \[
  e^{i
p(s)XP}P^2=e^{-2 p(s)}P^2 e^{ip(s)XP}.
\]
Thus
\begin{align*}
\frac{dR}{ds}&=\frac{1}{2}p^\prime (s)R(s)+q^\prime(s)X^2 R(s)
 -2 p^\prime (s)q(s)X^2 R(s) +ip^\prime (s) XP R(s)   \\
&+r^\prime(s)e^{-\frac{3}{2}p(s)}e^{q(s)X^2}P^2e^{i
p(s)XP}e^{r(s)P^2}.
\end{align*}
By (\ref{23}),
\[
 e^{q(s) X^2} P^2 =P^2  e^{q(s) X^2}  + \left( 2 q(s)-4 q(s)^2 X^2+4 i q(s) X
P\right)e^{q(s) X^2}.
\]
Thus, using (\ref{e2})-(\ref{e4}) to replace $p^\prime (s)$,
$q^\prime (s)$ and $r^\prime (s)$, we obtain
\[
\frac{dR}{ds}=(a X^2+b P^2) R(s).
\]
\end{proof}
\begin{lemma}\label{4} With  $\hbar=1$ and $X, P$  as in Section \ref{Heis}, for $n\in\mathbb{N}=\{1,2,...\}$,
\[
(XP)^n=\sum_{k=1}^n (-1)^{n-k} i^{n-k} S(n,k) X^k P^k,
\]
where  $S(n,k)$ are the Stirling numbers of the second kind.
\end{lemma}
\begin{proof}
 It is known (see, for example, \cite{NO}) that, if $\lbrack a, a^\dagger
\rbrack=1$ then
\[
 \left( a^\dagger a\right)^n=\sum_{k=1}^n S(n,k) {a^\dagger}^k a^k,
\]
from which the result follows by taking  $a=i P$ and
$a^\dagger=X$.
\end{proof}
Notice that for each $s=it$, $t\in\mathbb{R}$, using properties of
the hyperbolic functions and (\ref{e2})-(\ref{e4}), we see that if
$ab<0$ then both sides of (\ref{e1}) consist of unitary operators
(see, for example,  Chapter 2 of \cite{GS}) of the form $e^{i f(t)
T}$ where $T$ is an unbounded, densely defined, self-adjoint
operator and $f:\mathbb{R}\to \mathbb{R}$.  In particular, using
the $X, P$ commutation relations,
\[
e^{ip(it)XP}=e^{if(t) (XP+PX)}e^{-f(t)I},
\]
where $f(t)=\frac{p(it)}{2}\in\mathbb{R}$.
\begin{theorem}\label{wtf} In the notation of Lemma \ref{3}, with $s=it$,
$t\in\mathbb{R}$,
\[
\langle \Phi, e^{it (a X^2+b P^2)}\Phi \rangle=
\frac{\sqrt{2}\,e^{\frac{1}{2}p(it)}}{\sqrt{p(it)^2+(2
\,q(it)-1))(2 \,r(it)-1)}}.
\]
\end{theorem}
\begin{proof} By Lemma \ref{3},
\begin{align*}
\langle \Phi, e^{it (a X^2+b P^2)}\Phi \rangle &=\langle \Phi,
e^{\frac{1}{2}p(it)}e^{q(it)X^2}e^{i p(it) X P}e^{r(it)P^2} \Phi
\rangle\\
&= e^{\frac{1}{2}p(it)}\langle e^{\overline{q(it)}X^2}\Phi, e^{i
p(it) X P}e^{r(it)P^2} \Phi \rangle .
\end{align*}
Using
\[
e^{\overline{q(it)}X^2}=\int_{\mathbb{R}}e^{\overline{q(it)}\lambda^2}
\,dE_\lambda
\,\,;\,\,e^{r(it)P^2}=\int_{\mathbb{R}}e^{r(it)\lambda^2}
\,dE^{\prime}_\lambda ,
\]
where $\{E_\lambda\}$ and $\{E^{\prime}_\lambda\}$ are the
spectral resolutions of $X$ and $P$ respectively, and the fact that the inner product is linear in the second and conjugate linear in the first argument, we have
\begin{align*}
\langle \Phi, e^{it (a X^2+b P^2)}\Phi \rangle &=\langle \Phi,
e^{\frac{1}{2}p(it)}e^{q(it)X^2}e^{i p(it) X P}e^{r(it)P^2} \Phi
\rangle\\
&=e^{\frac{1}{2}p(it)}\int_{\mathbb{R}}\int_{\mathbb{R}}e^{q(it)\lambda^2}
e^{r(it){\lambda^\prime}^2}d_\lambda\,d_{\lambda^\prime} \langle
E_\lambda\Phi, e^{i p(it) X P}E^{\prime}_{\lambda^\prime} \Phi
\rangle .
\end{align*}
Using, as in the proof of Theorem \ref{1Dh},
\[
(E_{\lambda}\Phi)(t)=\Phi(t)\chi_{ { _{( -\infty, \lambda\rbrack
}}}(t),
\]
and
\[
(U^{-1}\Phi)(t)={\pi}^{-1/4}e^{-\frac{t^2}{2}}=\Phi(t),
\]
we obtain
\begin{align*}
\left(E^\prime_{\lambda}\Phi\right)(t)&=
\left(U E_\lambda U^{-1}\Phi\right)(t)=\left(U \chi_{ { _{( -\infty, \lambda\rbrack}}} \Phi \right)(t)\\
&=(2 \pi)^{-1/2}\,\int_{-\infty}^\infty
 \,e^{ist}\chi_{ { _{( -\infty, \lambda\rbrack }}}(s) \Phi(s)\,ds
=2^{-1/2}{\pi}^{-3/4}\int_{-\infty}^\lambda
e^{ist-\frac{s^2}{2}}\,ds .
\end{align*}
By Lemma \ref{4},
\begin{align*}
\langle E_\lambda\Phi&, e^{i p(it) X P}E^{\prime}_{\lambda^\prime} \Phi\rangle =\sum_{n=1}^\infty \frac{i^n p(it)^n}{n!}\langle E_\lambda\Phi,  (X P)^n E^{\prime}_{\lambda^\prime} \Phi\rangle\\
&=\sum_{n=1}^\infty \frac{i^n p(it)^n}{n!}\sum_{k=1}^n (-1)^{n-k} i^{n-k}S(n,k)\langle E_\lambda\Phi,  X^k P^k E^{\prime}_{\lambda^\prime} \Phi\rangle\\
&=\sum_{n=1}^\infty \sum_{k=1}^n (-1)^{n-k}\frac{i^{2n-k} p(it)^n}{n!}  S(n,k)\langle E_\lambda\Phi,  X^k P^k E^{\prime}_{\lambda^\prime} \Phi\rangle\\
&=\sum_{n=1}^\infty \sum_{k=1}^n (-1)^{n-k}\frac{i^{2n-k} p(it)^n}{n!}  S(n,k)\langle X^k E_\lambda\Phi,   P^k E^{\prime}_{\lambda^\prime} \Phi\rangle\\
& \sum_{n=1}^\infty \sum_{k=1}^n \frac{i^k p(it)^n}{n!}
S(n,k)\langle X^k E_\lambda\Phi,   P^k E^{\prime}_{\lambda^\prime}
\Phi\rangle .
\end{align*}
Since
\[
X^k=\int_{\mathbb{R}}\Lambda^k \,dE_\Lambda
\,\,\,\,;\,\,\,\,P^k=\int_{\mathbb{R}}M^k \,dE^{\prime}_M ,
\]
we have
\begin{align*}
\langle X^k E_\lambda\Phi&, P^k E^{\prime}_{\lambda^\prime} \Phi\rangle=\int_{\mathbb{R}}\int_{\mathbb{R}}\Lambda^k M^k \,d_\Lambda d_M \langle  E_\Lambda E_\lambda\Phi, E^{\prime}_M E^{\prime}_{\lambda^\prime} \Phi\rangle\\
&= \int_{\mathbb{R}}\int_{\mathbb{R}}\Lambda^k M^k \,d_\Lambda d_M \langle  E_{\min (\Lambda ,\lambda)}\Phi, E^{\prime}_{\min(M,\lambda^\prime)} \Phi\rangle\\
&= \int_{\mathbb{R}}\int_{\mathbb{R}}\Lambda^k M^k \,d_\Lambda d_M
 \left( \int_{\mathbb{R}}\overline{ E_{\min (\Lambda ,\lambda)}\Phi (t)}\, E^{\prime}_{\min(M,\lambda^\prime)} \Phi (t)\,dt\right)\\
&=\frac{1}{\pi
\sqrt{2}}\int_{\mathbb{R}}\int_{\mathbb{R}}\Lambda^k M^k
\,d_\Lambda d_M \left(\int_{-\infty}^{\min (\Lambda
,\lambda)}\int_{-\infty}^{\min (M,\lambda^\prime)}
e^{ist-\frac{1}{2}(s^2+t^2)}\,ds dt \right).
\end{align*}
If $\min (\Lambda ,\lambda)=\lambda$ and/or $\min (M,\lambda^\prime)= \lambda^\prime $ then
\[
d_\Lambda d_M \left(\int_{-\infty}^{\min (\Lambda
,\lambda)}\int_{-\infty}^{\min (M,\lambda^\prime)}
e^{ist-\frac{1}{2}(s^2+t^2)}\,ds dt\right)=0.
\]
Thus we get a nonzero result for $\min (\Lambda ,\lambda)=\Lambda$
and $\min (M,\lambda^\prime)= M$, i.e., for $\Lambda \leq \lambda$
and $M \leq \lambda^\prime$, in which case
\begin{align*}
\langle X^k E_\lambda\Phi, P^k E^{\prime}_{\lambda^\prime} \Phi\rangle&=\frac{1}{\pi \sqrt{2}}\int_{-\infty}^\lambda \int_{-\infty}^{\lambda^\prime} \Lambda^k M^k \,d_\Lambda d_M \left(\int_{-\infty}^\Lambda \int_{-\infty}^M e^{ist-\frac{1}{2}(s^2+t^2)}\,ds dt\right) \\
&=   \frac{1}{\pi \sqrt{2}} \int_{-\infty}^\lambda
\int_{-\infty}^{\lambda^\prime} \Lambda^k M^k e^{i M \Lambda-
\frac{1}{2}(M^2+\Lambda^2) }\,dM d\Lambda ,
\end{align*}
and
\[
d_\lambda\,d_{\lambda^\prime}\langle X^k E_\lambda\Phi, P^k
E^{\prime}_{\lambda^\prime} \Phi\rangle=\frac{1}{\pi \sqrt{2}}
\lambda^k{\lambda^\prime}^k e^{i \lambda \lambda^\prime-
\frac{1}{2}(\lambda^2+{\lambda^\prime}^2) }\, d\lambda
d\lambda^\prime .
\]
Thus
\begin{align*}
d_\lambda\,d_{\lambda^\prime}\langle E_\lambda\Phi&, e^{i p(it) X P}E^{\prime}_{\lambda^\prime} \Phi\rangle=\sum_{n=1}^\infty \sum_{k=1}^n \frac{i^{k} p(it)^n}{n!}  S(n,k)d_\lambda\,d_{\lambda^\prime}\langle X^k E_\lambda\Phi,   P^k E^{\prime}_{\lambda^\prime} \Phi\rangle\\
&=\frac{1}{\pi \sqrt{2}}\sum_{n=1}^\infty \sum_{k=1}^n \frac{i^{k}
p(it)^n}{n!}  S(n,k) \lambda^k{\lambda^\prime}^k e^{i \lambda
\lambda^\prime- \frac{1}{2}(\lambda^2+{\lambda^\prime}^2) }\,
d\lambda   d\lambda^\prime ,
\end{align*}
and
\begin{align*}
\langle \Phi, e^{it (a X^2+b P^2)}\Phi \rangle
 &=\frac{1}{\pi \sqrt{2}}e^{\frac{1}{2}p(it)}\sum_{n=1}^\infty \frac{ p(it)^n}{n!}\\&\cdot \int_{\mathbb{R}}\int_{\mathbb{R}} e^{q(it)\lambda^2}
e^{r(it){\lambda^\prime}^2} \sum_{k=1}^n  i^{k} S(n,k)
\lambda^k{\lambda^\prime}^k e^{i \lambda \lambda^\prime-
\frac{1}{2}(\lambda^2+{\lambda^\prime}^2) }\, d\lambda
d\lambda^\prime .
\end{align*}
Since $S(n,0)=0$ and $S(n,k)=0$ for $k>n$,
\[
\sum_{n=1}^\infty\sum_{k=1}^n S(n,k)\frac{ p(it)^n}{n!}{(i
\lambda\lambda^\prime)}^k=\sum_{k=0}^\infty\sum_{n=0}^\infty
S(n,k)\frac{ p(it)^n}{n!}{(i \lambda\lambda^\prime)}^k=e^{i
\lambda\lambda^\prime\left(e^{p(it)}-1\right)},
\]
where we have used the identity (see \cite{S}, equation (9.70)),
\[
\sum_{n=0}^\infty \frac{ x^n}{n!}
S(n,k)=\frac{\left(e^{x}-1\right)^k }{k!}.
\]
Thus
\[
\langle \Phi, e^{it (a X^2+b P^2)}\Phi \rangle=\frac{1}{\pi
\sqrt{2}}e^{\frac{1}{2}p(it)}\int_{\mathbb{R}}\int_{\mathbb{R}}
e^{\left(q(it)-\frac{1}{2}\right)\lambda^2+\left(r(it)-\frac{1}{2}\right){\lambda^\prime}^2+i\lambda
\lambda^\prime e^{p(it)}}\, d\lambda   d\lambda^\prime .
\]
Using the integration formula
\[
\int_{\mathbb{R}}\int_{\mathbb{R}} e^{\alpha x^2+\beta y^2+i\gamma
x y}\, dx dy=\frac{2\pi}{\sqrt{\gamma^2+4 \alpha \beta}},
\]
we obtain
\[
\langle \Phi, e^{it (a X^2+b P^2)}\Phi
\rangle=\frac{\sqrt{2}\,e^{\frac{1}{2}p(it)}}{\sqrt{p(it)^2+(2
\,q(it)-1))(2 \,r(it)-1)}}.
\]
\end{proof}
\begin{corollary} For  $\hbar=1$,
the vacuum characteristic function of the quantum observable
\[
H=K_+ + K_- + K_0
\]
is
\[
\langle \Phi, e^{itH} \Phi \rangle=\left( \frac{6 \,\rm{sech}
\left(\frac{t\sqrt{3}}{2}\right) }{3+3 \log^2 \left(\rm{sech}
\left(\frac{t\sqrt{3}}{2} \right)\right)-3 \tanh^2
\left(\frac{t\sqrt{3}}{2} \right) + 4 i \sqrt{3} \tanh
\left(\frac{t\sqrt{3}}{2} \right) } \right)^{\frac{1}{2}}.
\]
\end{corollary}
\begin{proof} For $X, P$ and $\Phi$ as in Section \ref{Heis},
\[
H=K_+ + K_- + K_0=\frac{1}{4}\left(3 X^2-P^2\right)+\frac{1}{2}I \
.
\]
Thus, by Theorem \ref{wtf} and Lemma \ref{3} with $a=\frac{3}{4}$
and $b=-\frac{1}{4}$, using
\[
\tanh(-x)=-\tanh (x)\,\,;\,\,\rm{sech} (-x)=\rm{sech} (x),
\]
 we have
\begin{align*}
q(it)&= -\frac{i\sqrt{3}}{2} \tanh \left(\frac{t\sqrt{3}}{2} \right), \\
p(it)&=\log \left(\rm{sech} \left(\frac{t\sqrt{3}}{2} \right)\right), \\
r(it)&=-\frac{i\sqrt{3}}{6} \tanh \left(\frac{t\sqrt{3}}{2}
\right),
\end{align*}
and
\begin{align*}
\langle \Phi, e^{itH} \Phi
\rangle&=\frac{\sqrt{2}\,e^{\frac{1}{2}p(it)}}{\sqrt{p(it)^2+(2
\,q(it)-1))(2
\,r(it)-1)}}\\
&=\left( \frac{6 \,\rm{sech} \left(\frac{t\sqrt{3}}{2}\right)
}{3+3 \log^2 \left(\rm{sech} \left(\frac{t\sqrt{3}}{2}
\right)\right)-3 \tanh^2 \left(\frac{t\sqrt{3}}{2} \right) + 4 i
\sqrt{3} \tanh \left(\frac{t\sqrt{3}}{2} \right) }
\right)^{\frac{1}{2}}.
\end{align*}
\end{proof}
\section{Computing the Vacuum Resolution and the Characteristic Function with Stone's Formula}\label{Stone}
In the following  we  use  Stone's Formula (\ref{srv}) to compute
the vacuum resolution of the identity, i.e., $\langle \Phi,
E_\lambda \Phi\rangle$, and the characteristic function
\[
\int_{\mathbb{R}}e^{it\lambda}\,d\langle\Phi, E_{\lambda}\Phi
\rangle,
\]
of the operators $X, P$ and $X+P$ of Section \ref{Heis}.
\begin{theorem}\label{X} The vacuum spectral resolution of $X$ is
\[
\langle \Phi, E_\lambda \Phi\rangle= \frac{1}{\pi^{1/2} }
\int_{-\infty}^{\lambda} e^{-s^2}\,ds .
\]
Moreover, for $t\in \mathbb{R}$,
\[
\langle \Phi, e^{i t X} \Phi \rangle=\sqrt{2}\,
e^{-\frac{t^2}{2}}.
\]
\end{theorem}
\begin{proof} For $a\in\mathbb{C}$ with ${\rm Im}\,a\neq 0$ and for $s\in\mathbb{R}$,
\begin{align*}
R(a;X)g(s)=G(s)& \iff g(s)=(a-X)G(s)\\
& \iff g(s)=(a-s)G(s)\\
& \iff G(s)=\frac{g(s)}{a-s}.
\end{align*}
The function $G$ is in the domain of $X$, since the continuous
functions
\[
\psi(s)=\frac{1}{|a-s|^2}\,;\,\sigma(s)=s^2 \psi(s)
\]
satisfy
\[
\lim_{s\to\pm\infty}\psi(s)=0\in \mathbb{R},
\lim_{s\to\pm\infty}\psi(s)=1\in \mathbb{R},
\]
and are therefore bounded. Thus, since $g\in
L^2(\mathbb{R},\mathbb{C})$,
\[
\int_{\mathbb{R}}|G(s)|^2\,ds<+\infty,
\int_{\mathbb{R}}s^2\,|G(s)|^2\,ds<+\infty .
\]
 By (\ref{srv}),
\begin{align*}
&\langle \Phi, E_\lambda \Phi\rangle= \lim_{\rho\to
0^+}\lim_{\delta\to 0^+}\lim_{\epsilon\to 0^+}\frac{1}{2\pi
i}\int_{-\infty}^{\lambda+\rho-\delta}\langle\Phi,
\left(R(t-\epsilon i; X)-R(t+\epsilon i; X)\right)\Phi\rangle\,dt\\
&=\lim_{\rho\to 0^+}\lim_{\delta\to 0^+}\lim_{\epsilon\to 0^+}
\frac{1}{2\pi
i}\int_{-\infty}^{\lambda+\rho-\delta}\int_{-\infty}^{\infty}\overline{\Phi(s)}
\left(R(t-\epsilon i; X)-R(t+\epsilon i; X)\right)\Phi(s)\,ds\,dt\\
&= \lim_{\rho\to 0^+}\lim_{\delta\to 0^+}\lim_{\epsilon\to
0^+}\frac{1}{2\pi
i}\int_{-\infty}^{\lambda+\rho-\delta}\int_{-\infty}^{\infty}\overline{\Phi(s)}\Phi(s)
\left( \frac{1}{t-i\epsilon-s}-
\frac{1}{t+i\epsilon-s}\right)\,ds\,dt\\
&=\lim_{\epsilon\to 0^+}\frac{\epsilon}{\pi^{3/2}
}\int_{-\infty}^{\lambda}\int_{-\infty}^{\infty}
\frac{e^{-s^2}}{(t-s)^2+\epsilon^2}\,ds\,dt.
\end{align*}
Since the continuous (thus measurable) function
\[
(t,s)\in ( -\infty, \lambda\rbrack \times \mathbb{R} \to
\frac{e^{-s^2}}{(t-s)^2+\epsilon^2}\in\mathbb{R}
\]
is nonnegative, by Fubini's theorem we can reverse the order of
integration and obtain
\[
\langle \Phi, E_\lambda \Phi\rangle =\lim_{\epsilon\to
0^+}\frac{\epsilon}{\pi^{3/2}
}\int_{-\infty}^{\infty}e^{-s^2}\int_{-\infty}^{\lambda}
\frac{1}{(t-s)^2+\epsilon^2}\,dt\,ds,
\]
which, since
\begin{align*}
\int_{-\infty}^{\lambda}
\frac{1}{(t-s)^2+\epsilon^2}\,dt&=-\frac{1}{\epsilon}\arctan
\left(\frac{s-t}{\epsilon}
\right)\,|^{t=\lambda}_{t=-\infty}\\
&=-\frac{1}{\epsilon}\left(\arctan
\left(\frac{s-\lambda}{\epsilon} \right)-\frac{\pi}{2} \right),
\end{align*}
becomes
\[
\langle \Phi, E_\lambda \Phi\rangle =\lim_{\epsilon\to
0^+}\frac{1}{\pi^{3/2}
}\int_{-\infty}^{\infty}e^{-s^2}\left(\frac{\pi}{2}-\arctan
\left(\frac{s-\lambda}{\epsilon} \right) \right) \,ds .
\]
Splitting the integral as
\[
\int_{-\infty}^{\infty}=\int_{-\infty}^{\lambda}+\int_{\lambda}^\infty
,
\]
we notice that
\[
\lim_{\epsilon\to 0^+}\arctan\left(\frac{s-\lambda}{\epsilon}
\right)=\pm\frac{\pi}{2},
\]
with the minus and plus signs corresponding to the first and
second  integral respectively. Thus, using the bounded convergence
theorem to pass  $\lim_{\epsilon\to 0^+}$ under the integral sign,
we obtain
\[
\langle \Phi, E_\lambda \Phi\rangle =\frac{1}{\pi^{1/2}
}\int_{-\infty}^{\lambda}e^{-s^2} \,ds.
\]
Thus,
\begin{align*}
\langle \Phi, e^{i t X}\Phi
\rangle&=\int_{\mathbb{R}}e^{it\lambda}\,d\langle\Phi,
E_{\lambda}\Phi \rangle\\
&=\int_{\mathbb{R}}e^{it\lambda}\,\frac{1}{\sqrt{\pi}}
e^{-\frac{\lambda^2}{2}
} \,d\lambda\\
&=\frac{1}{\sqrt{\pi}}\int_{\mathbb{R}}e^{it\lambda-\frac{\lambda^2}{2}
} \,d\lambda\\
&=\sqrt{2}\, e^{-\frac{t^2}{2}}.
\end{align*}
\end{proof}
\begin{theorem}\label{P} The vacuum spectral resolution of $P$ is
\[
\langle \Phi, E_\lambda \Phi\rangle= \pi^{-\frac{1}{2}}
\int_{-\infty}^{\lambda} e^{-s^2}\,ds.
\]
Moreover, for $t\in \mathbb{R}$,
\[
\langle \Phi, e^{i t P} \Phi \rangle=\sqrt{2}\,
e^{-\frac{t^2}{2}}.
\]
\end{theorem}
\begin{proof} For $a\in\mathbb{C}$ with ${\rm Im}\,a\neq 0$ and for $s\in\mathbb{R}$,
\begin{align}
R(a;P)g(s)=G(s)& \iff g(s)=(a-P)G(s)\notag\\
& \iff g(s)=\left(a+i \frac{d}{ds} \right)G(s)\notag\\
& \iff G^\prime(s)-i a G(s)=-i g(s)\notag \\
& \iff \left( e^{-ias} G(s) \right)^\prime
=-ie^{-ias}g(s)\label{tbi}.
\end{align}
Since $G$ is in the domain of $P$, it follows that  $\lim_{t\to
\pm\infty}G(t)=0$. Therefore,
\[
\lim_{t\to +\infty}|e^{-iat} G(t)|=\lim_{t\to +\infty}e^{t\, {\rm
Ima}} |G(t)|=0, \mbox{ if   }{\rm Im}\,a<0 ,
\]
and
\[
\lim_{t\to -\infty}|e^{-iat} G(t)|=\lim_{t\to -\infty}e^{t\, {\rm
Ima}} |G(t)|=0, \mbox{ if   }{\rm Im}\,a> 0.
\]
For ${\rm Im} a<0$, integrating (\ref{tbi}) from $t$ to $\infty$
and then replacing $t$ by $s$,  we obtain
\[
-e^{-as}G(s)=-i\int^{\infty}_s e^{ia(s-t)}g(t) \,dt,
\]
so
\[
G(s)=i\int^{\infty}_s e^{ia(s-t)}g(t) \,dt,
\]
while, for ${\rm Im}\,a> 0$, integrating (\ref{tbi}) from
$-\infty$ to $t$  and then replacing $t$ by $s$,  we obtain
\[
G(s)=-i\int_{-\infty}^s e^{ia(s-t)}g(t) \,dt .
\]
Thus,
\[
R(a;P)g(s)=\left\{
\begin{array}{llr}
-i\int_{-\infty}^s e^{ia(s-t)}g(t) \,dt  , &\mbox{ if  } {\rm Im}\,a> 0   \\
& \\
 i\int^{\infty}_s e^{ia(s-t)}g(t) \,dt ,&\mbox{ if  } {\rm
Im}\, a<0
\end{array}
\right. .
\]
For $g(t)=\Phi(t)$ we find
\[
R(a;P)\Phi(s)=\left\{
\begin{array}{llr}
-i
\pi^{-1/4}\int_{-\infty}^s e^{ia(s-t)-\frac{t^2}{2}} \,dt  ,&\mbox{ if  } {\rm Im}\, a> 0   \\
& \\
i \pi^{-1/4}\int^{\infty}_s e^{ia(s-t)-\frac{t^2}{2}} \,dt,
&\mbox{ if  } {\rm Im}\,a<0
\end{array}
\right. .
\]
Thus,
\begin{align*}
&\langle \Phi, E_\lambda \Phi\rangle= \lim_{\rho\to
0^+}\lim_{\delta\to 0^+}\lim_{\epsilon\to 0^+}\frac{1}{2\pi
i}\int_{-\infty}^{\lambda+\rho-\delta}\langle\Phi,
\left(R(t-\epsilon i; P)-R(t+\epsilon i; P)\right)\Phi\rangle\,dt\\
&=\lim_{\epsilon\to 0^+} \frac{1}{2\pi
i}\int_{-\infty}^{\lambda}\int_{-\infty}^{\infty}\overline{\Phi(s)}
\left(R(t-\epsilon i; P)-R(t+\epsilon i; P)\right)\Phi(s)\,ds\,dt\\
&= \lim_{\epsilon\to
0^+}\frac{1}{2\pi^{3/2}}\int_{-\infty}^{\lambda}
\int_{-\infty}^{\infty}e^{-\frac{s^2}{2}}\cdot\\
&\cdot\left(e^{(\epsilon+i t)s}\int_s^{\infty} e^{-(i
t+\epsilon)w- \frac{w^2}{2}  }\,dw+ e^{(i
t-\epsilon)s}\int_{-\infty}^s e^{(\epsilon-i t)w- \frac{w^2}{2}
}\,dw\right)\, ds\,dt.
\end{align*}
For each $t\in (-\infty, \lambda \rbrack$, using the triangle
inequality, the integration formulas
\[
\int_{-\infty}^\infty e^{\pm \epsilon w- \frac{w^2}{2}
}\,dw=\sqrt{2 \pi} e^{\frac{\epsilon^2 }{2}},
\int_{-\infty}^\infty e^{- \frac{s^2}{2} }\,ds=\sqrt{2 \pi},
\]
 and the fact that, for sufficiently small $\epsilon$,
 $e^{\frac{\epsilon^2 }{2}}<3$ and $e^{ \epsilon s}+ e^{-\epsilon s}<3$,
  we see that the function
\[
f_\epsilon (s; t):=e^{-\frac{s^2}{2}}\cdot\left(e^{(\epsilon+i
t)s}\int_s^{\infty} e^{-(i t+\epsilon)w- \frac{w^2}{2}  }\,dw+
e^{(i t-\epsilon)s}\int_{-\infty}^s e^{(\epsilon-i t)w-
\frac{w^2}{2} }\,dw\right),
\]
satisfies
\[
|f_\epsilon (s; t)|\leq f(s):=9 \sqrt{2 \pi} e^{-\frac{s^2}{2}},
\]
with
\[
\int_{-\infty}^\infty f(s)\,ds = 18 \pi <\infty .
\]
Thus, by the Bounded Convergence Theorem, we may pass the limit
$\epsilon \to 0^+$ inside the integral with respect to $s$.
Moreover, for $a<0$ and $t\in (-a, \lambda\rbrack $,
\[
| \int_{-\infty}^\infty f_\epsilon (s; t)\,ds | \leq 18 \pi ,
\]
with
\[
\int_{a}^{\lambda}  18 \pi \,dt=(\lambda-a)18 \pi <\infty .
\]
Thus, by the Bounded Convergence Theorem, we may pass the limit
$\epsilon \to 0^+$ inside  $\int_{a}^{\lambda}(...)\,dt$ as well,
and we obtain
\begin{align*}
\langle \Phi, E_\lambda \Phi\rangle&= \lim_{a\to -\infty}\langle
\Phi, E\left((a,\lambda \rbrack\right) \Phi\rangle\\
&=\lim_{a\to -\infty}\lim_{\epsilon\to
0^+}\frac{1}{2\pi^{3/2}}\int_{a}^{\lambda}
\int_{-\infty}^{\infty}e^{-\frac{s^2}{2}}\cdot\\
&\cdot\left(e^{(\epsilon+i t)s}\int_s^{\infty} e^{-(i
t+\epsilon)w- \frac{w^2}{2}  }\,dw+ e^{(i
t-\epsilon)s}\int_{-\infty}^s e^{(\epsilon-i t)w- \frac{w^2}{2}
}\,dw\right)\, ds\,dt \\
&=\lim_{a\to -\infty}\frac{1}{2\pi^{3/2}}\int_{a}^{\lambda}
\int_{-\infty}^{\infty}e^{-\frac{s^2}{2}}\cdot\\
&\cdot\left(e^{i t s}\int_s^{\infty} e^{-i t w- \frac{w^2}{2}
}\,dw+ e^{i ts}\int_{-\infty}^s e^{-i t w- \frac{w^2}{2}
}\,dw\right)\, ds\,dt\\
&= \frac{1}{2\pi^{3/2}}\int_{-\infty}^{\lambda}
\int_{-\infty}^{\infty}e^{-\frac{s^2}{2}+i t s}
\int_{-\infty}^{\infty} e^{-i tw- \frac{w^2}{2} }\,dw  \, ds\,dt
\\
&=\frac{1}{2\pi^{3/2}}\int_{-\infty}^{\lambda}
\int_{-\infty}^{\infty}e^{-\frac{s^2}{2}+i t s}\, ds\,
\int_{-\infty}^{\infty} e^{-i tw- \frac{w^2}{2} }\,dw  \,dt\\
&=\frac{1}{2\pi^{3/2}}\int_{-\infty}^{\lambda}\left(\sqrt{2
\pi}e^{-\frac{t^2}{2}} \right) \left( \sqrt{2
\pi}e^{-\frac{t^2}{2} } \right)\,dt\\
&=  \pi^{-\frac{1}{2}}  \int_{-\infty}^{\lambda}e^{-t^2} \,dt.
\end{align*}
Thus,
\begin{align*}
\langle \Phi, e^{i t P}\Phi
\rangle&=\int_{\mathbb{R}}e^{it\lambda}\,d\langle\Phi,
E_{\lambda}\Phi
\rangle=\int_{\mathbb{R}}e^{it\lambda}\,\frac{1}{\sqrt{\pi}}
e^{-\frac{\lambda^2}{2}
} \,d\lambda\\
&=\frac{1}{\sqrt{\pi}}\int_{\mathbb{R}}e^{it\lambda-\frac{\lambda^2}{2}
} \,d\lambda=\sqrt{2}\, e^{-\frac{t^2}{2}}.
\end{align*}
\end{proof}
\begin{theorem}\label{X+P} The vacuum spectral resolution of $X+P$ is
\[
\langle \Phi, E_\lambda
\Phi\rangle=\frac{1}{\sqrt{2\pi}}\int_{-\infty}^{\lambda}
e^{-\frac{t^2}{2} } \,dt.
\]
Moreover, for $t\in \mathbb{R}$,
\[
\langle \Phi, e^{i t (X+P)} \Phi \rangle=e^{-\frac{t^2}{2}}.
\]
\end{theorem}
\begin{proof} For $a\in\mathbb{C}$ with ${\rm Im}\,a\neq 0$ and for $s\in\mathbb{R}$,
\begin{align}
R(a;X+P)g(s)=G(s)& \iff g(s)=(a-X-P)G(s)\notag\\
& \iff g(s)=\left(a-s+i \frac{d}{ds} \right)G(s)\notag\\
& \iff G^\prime(s)+i(s- a) G(s)=-i g(s)\notag \\
& \iff \left( e^{i\left(\frac{s^2}{2}-as\right)} G(s)
\right)^\prime =-i e^{i\left(\frac{s^2}{2}-as\right)}
g(s)\label{tbii}.
\end{align}
As in the proof of Theorems \ref{X} and \ref{P}, for ${\rm
Im}\,a<0$, integrating (\ref{tbii}) from $t$ to $\infty$ and then
replacing $t$ by $s$,  we obtain
\[
-e^{i\left(\frac{s^2}{2}-as\right)} G(s)=-i\int^{\infty}_s
e^{i\left(\frac{t^2}{2}-at\right)}g(t) \,dt,
\]
so
\[
G(s)=i\int^{\infty}_s
e^{i\left(\frac{t^2-s^2}{2}-a(t-s)\right)}g(t) \,dt,
\]
while, for ${\rm Im}\,a> 0$, integrating (\ref{tbii}) from
$-\infty$ to $t$  and then replacing $t$ by $s$,  we obtain
\[
G(s)=-i\int_{-\infty}^s e^{i\left(\frac{t^2-s^2}{2}-a(t-s)\right)}
g(t) \,dt.
\]
Thus,
\[
R(a;X+P)g(s)=\left\{
\begin{array}{llr}
 -i\int_{-\infty}^s e^{i\left(\frac{t^2-s^2}{2}-a(t-s)\right)}
g(t) \,dt , &\mbox{ if  } {\rm Im}\, a> 0 \\
& \\
i\int^{\infty}_s e^{i\left(\frac{t^2-s^2}{2}-a(t-s)\right)}g(t)
\,dt ,&\mbox{ if  } {\rm Im}\, a<0
\end{array}
\right.  .
\]
For $g(t)=\Phi(t)$ we find
\[
R(a; X+P)\Phi(s)=\left\{
\begin{array}{llr}
-i
\pi^{-1/4}\int_{-\infty}^s  e^{i\left(\frac{t^2-s^2}{2}-a(t-s)\right)-\frac{t^2}{2}}   \,dt  ,&\mbox{ if  } {\rm Im}\,a> 0   \\
& \\
i \pi^{-1/4}\int^{\infty}_s
e^{i\left(\frac{t^2-s^2}{2}-a(t-s)\right)-\frac{t^2}{2}}     \,dt\
, &\mbox{ if } {\rm Im}\,a<0
\end{array}
\right. .
\]
Thus, as in the proof of Theorems \ref{X} and \ref{P},
\begin{align*}
&\langle \Phi, E_\lambda \Phi\rangle\\
&= \lim_{\epsilon\to 0^+} \frac{1}{2\pi
i}\int_{-\infty}^{\lambda}\int_{-\infty}^{\infty}\overline{\Phi(s)}
\left(R(t-\epsilon i; X+P)-R(t+\epsilon i; X+P)\right)\Phi(s)\,ds\,dt\\
&= \lim_{\epsilon\to
0^+}\frac{1}{2\pi^{3/2}}\int_{-\infty}^{\lambda}
\int_{-\infty}^{\infty}e^{-\frac{s^2}{2}}\cdot\\
&\cdot\left(\int_s^{\infty} e^{i\left(\frac{w^2-s^2}{2}-(t-i
\epsilon)(w-s)\right)- \frac{w^2}{2} }\,dw+ \int_{-\infty}^s
e^{i\left(\frac{w^2-s^2}{2}-(t+i \epsilon)(w-s)\right)-
\frac{w^2}{2} }\,dw\right)\, ds\,dt\\
&=\frac{1}{2\pi^{3/2}}\int_{-\infty}^{\lambda}
\int_{-\infty}^{\infty}e^{-\frac{s^2}{2}}\int_{-\infty}^{\infty}
e^{i\left(\frac{w^2-s^2}{2}-t(w-s)\right)- \frac{w^2}{2}
}\,dw\,ds\,dt\\
&=\frac{1}{2\pi^{3/2}}\int_{-\infty}^{\lambda} \left(
\int_{-\infty}^{\infty} e^{-(1+i)\frac{s^2}{2}+i ts } \, ds
\right)\left(\int_{-\infty}^{\infty}
 e^{-(1-i)\frac{w^2}{2}-i t w}\, dw\right)\,dt\\
 &=\frac{1}{2\pi^{3/2}}\int_{-\infty}^{\lambda} \left|
\int_{-\infty}^{\infty} e^{-(1+i)\frac{s^2}{2}+i ts } \, ds
\right|^2\,dt\\
&=\frac{1}{2\pi^{3/2}}\int_{-\infty}^{\lambda} \left|
\sqrt{\frac{2\pi}{1+i} }e^{-\frac{1-i}{4} t^2  }
\right|^2\,dt\\
&=\frac{1}{2\pi^{3/2}}\int_{-\infty}^{\lambda} \sqrt{2} \pi
e^{-\frac{t^2}{2}  }
\,dt\\
&=\frac{1}{\sqrt{2\pi}}\int_{-\infty}^{\lambda} e^{-\frac{t^2}{2}
} \,dt.
\end{align*}
Thus,
\begin{align*}
\langle \Phi, e^{i t(X+P)}\Phi
\rangle&=\int_{\mathbb{R}}e^{it\lambda}\,d\langle\Phi,
E_{\lambda}\Phi
\rangle=\int_{\mathbb{R}}e^{it\lambda}\,\frac{1}{\sqrt{2\pi}}
e^{-\frac{\lambda^2}{2}
} \,d\lambda\\
&=\frac{1}{\sqrt{2\pi}}\int_{\mathbb{R}}e^{it\lambda-\frac{\lambda^2}{2}
} \,d\lambda=e^{-\frac{t^2}{2}}.
\end{align*}
\end{proof}
\begin{remark} \rm Computing the resolvent and the resulting  integral in
Stone's formula for the \textit{anti-commutator} operator
$T=XP+PX$ is not easy. Not much can be found on that in the
literature. The operator $T=XP+PX$ is a case where the Lie
algebraic method of using an appropriate splitting lemma, seems to
have an advantage over the analytic method that uses the vacuum
spectral resolution. We will return to the computation of the
vacuum spectral resolution of $XP+PX$, as well as of $aX^2+bP^2$,
with the use of Stone's formula, in the sequel to this paper.
\end{remark}
The following Theorem gives an example of how the vacuum spectral
resolution can be computed, once the characteristic function is
known.
\begin{theorem}\label{XP+PX} For $t\in \mathbb{R}$,
\[
\langle \Phi, e^{i t (XP+PX)} \Phi \rangle=\left({\rm sech} \,t
\right)^{1/2} .
\]
Moreover, the differential with respect to $\lambda$ of the vacuum
spectral resolution of $XP+PX$ is
\begin{align*}
d\,\langle \Phi, E_\lambda
\Phi\rangle=\frac{1-i}{4\pi}\,\left(e^{-\frac{\pi\lambda}{2}}\,B\left(-1;
\frac{1-2 i \lambda}{4},\frac{1}{2}
\right)+e^{\frac{\pi\lambda}{2}}\,B\left(-1; \frac{1+2 i
\lambda}{4},\frac{1}{2} \right)\right)\,d\lambda ,
\end{align*}
where
\[
B(z; a, b)=\int_0^z t^{a-1} (1-t)^{b-1}\,dt
\]
is the \textit{incomplete Beta function}.
\end{theorem}
\begin{proof} Using
\[
X=\frac{a+a^{\dagger}}{\sqrt{2}}
\,\,,\,\,P=\frac{a-a^{\dagger}}{\sqrt{2}i}
\]
and
\[
\lbrack a, a^{\dagger}\rbrack=\mathbf{1},
\]
we find that
\[
XP+PX=i\left( (a^{\dagger})^2-a^2  \right).
\]
By Proposition 3.4 of \cite{AccBouCOSA}, (or by Proposition 3.9 of
\cite{AccBouArxiv}, see also Proposition 4.1.1 of \cite{Fein} and
Proposition 4 of \cite{AccBouPhys}, where it is shown that $XP+PX$
is a continuous binomial or Beta process), it follows that
\[
\langle \Phi, e^{i t (XP+PX)} \Phi \rangle=\langle \Phi, e^{ t
\left( a^2- (a^{\dagger})^2 \right)} \Phi \rangle =\left({\rm
sech} \,t \right)^{1/2}.
\]
Therefore,
\[
\frac{1}{\sqrt{2\pi}}\int_{\mathbb{R}}e^{it\lambda}\,d\langle\Phi,
E_{\lambda}\Phi \rangle
=\frac{1}{\sqrt{2\pi}}\int_{\mathbb{R}}e^{it\lambda}\,\frac{d}{d\lambda}\langle\Phi,
E_{\lambda}\Phi \rangle\,d\lambda=\frac{1}{\sqrt{2\pi}}\left({\rm
sech} \,t \right)^{1/2},
\]
which means that
\[
\frac{d}{d\lambda}\langle\Phi, E_{\lambda}\Phi \rangle
\]
is  the inverse Fourier transform of
\[
\frac{1}{\sqrt{2\pi}}\left({\rm sech} \,t \right)^{1/2},
\]
i.e.,
\begin{align*}
\frac{d}{d\lambda}\langle \Phi, E_\lambda
\Phi\rangle=\frac{1-i}{4\pi}\,\left(e^{-\frac{\pi\lambda}{2}}\,B\left(-1;
\frac{1-2 i \lambda}{4},\frac{1}{2}
\right)+e^{\frac{\pi\lambda}{2}}\,B\left(-1; \frac{1+2 i
\lambda}{4},\frac{1}{2} \right)\right).
\end{align*}
\end{proof}

\bibliographystyle{amsplain}

\begin{thebibliography}{99}

\bibitem{AccBouCOSA}
Accardi, L., Boukas, A. :  On the  characteristic function of
random variables associated with Boson Lie algebras, {\em
Communications on Stochastic  Analysis} {\bf 4} (2010), no.~ 4 ,
493--504.

\bibitem{AccBouArxiv}
\bysame:  Normally ordered disentanglement of multi-dimensional
Schr\"odinger algebra exponentials, {\em Communications on
Stochastic Analysis}  {\bf 12} (2018), no.~ 3,  283--328.

\bibitem{AccBouPhys}
\bysame:   Fock representation of the
 renormalized higher powers of white noise and
  the centerless Virasoro (or  Witt) Zamolodchikov
  $w_\infty$ *--Lie algebra,
  {\em J. Phys. A: Math. Theor.} {\bf 41} (2008), 1--12.

\bibitem{NO}
Blasiak, P., Horzela, A., Penson, K. A., Solomon, A. I., Duchamp,
G.H.E. :  Combinatorics and Boson normal ordering: A gentle
introduction,  {\em American Journal of Physics} {\bf 75} (2007),
no.~ 7,  639–-646 .

\bibitem{DS} Dunford, N., Schwartz, J.T. :{\em Linear Operator, Part
II: Spectral Theory,  Self Adjoint Operators in Hilbert Space}, J.
Wiley \& Sons, New York, 1963.

\bibitem{Fein2}
Feinsilver, P. J. :  Krawtchouk-Griffiths systems I: matrix
approach, {\em Communications on Stochastic Analysis} {\bf 10}
(2016), no.~3 , 297--320 .

\bibitem{Fein}  Feinsilver, P. J.,  Schott, R.: {\em Algebraic structures and operator calculus. Volumes I and III}, Kluwer,  1993.

\bibitem{QMI} Galindo, A., Pascual, P. : {\em Quantum Mechanics
I}, Springer-Verlag, Texts and Monographs in Physics,  1990.


\bibitem{S} Gould, H.W., Quaintance, J.: {\em Combinatorial Identities for Stirling Numbers}, World Scientific, 2016.

\bibitem{GS} Gustafson, S.J., Sigal, I.M. : {\em Mathematical Concepts of Quantum Mechanics}, Springer-Verlag, Universitext,  2003.

\bibitem{Hall} Hall, B. C. : {\em Lie groups, Lie algebras, and representations:
 An Elementary Introduction}, Springer, Graduate Texts in Mathematics no.~ 222, 2003.
Second Edition,  2015.

\bibitem{Hum}
Humphreys, J. E.: {\em Introduction to Lie Algebras and
Representation Theory}, Springer-Verlag, Graduate Texts in
Mathematics, vol. 9, 1972.

\bibitem{Brazil}
Novaes, M.: Some basics of $\mathfrak{su}(1,1)$, {\em Revista
Brasileira de Ensino de Fisica} {\bf 26} (2004), no.~4, 351--357.

\bibitem{Pa}
Parthasarathy, K. R.: {\em An introduction to quantum stochastic
calculus}, Birkhauser Boston Inc., 1992.

\bibitem{Richtmyer}
Richtmyer, R. D.: {\em Principles of Advanced Mathematical
Physics}, vol.1, Texts and Monographs in Physics, Springer-Verlag,
1978.

\bibitem{Ro} Roach, G.F. : {\em Wave Scattering by Time--Dependent
Perturbations}, Princeton Series in Applied Mathematics, Princeton
University Press, 2007.

\bibitem{Taylor}
Taylor, A. E., Lay, D. C.: {\em Introduction to Functional
Analysis}, Robert E. Krieger Publishing Company, 1986.


\bibitem{Yosida} Yosida, K.: {\em Functional Analysis},
Springer-Verlag, 6th ed., 1980.

\end{thebibliography}

\end{document}